\documentclass[review]{elsarticle}
\usepackage{amsmath,amsthm}
\usepackage{dutchcal}
\usepackage[margin=1.2 in]{geometry}
\usepackage{lineno,xcolor,float}
\theoremstyle{definition}

\usepackage[T1]{fontenc}
\newtheorem{theorem}{Theorem}

\usepackage{subcaption}

\usepackage{amsfonts,mathrsfs}
\usepackage{physics}
\usepackage{csquotes}
\usepackage[T1]{fontenc}
\usepackage{subcaption}
\usepackage{comment}
\usepackage{url}
\definecolor{newcolor}{rgb}{.8,.349,.1}
\usepackage{latexsym}
\usepackage{framed,multirow}
\usepackage{tocbibind}
\usepackage{setspace}
\usepackage{graphicx}
\usepackage{caption}
\usepackage{dcolumn}
\usepackage{bm}
\usepackage{fontenc}
\usepackage{float}
\usepackage{array}
\usepackage{tabularx}
\usepackage{booktabs}
\usepackage{mathtools}
\usepackage{amssymb}
\usepackage{bm}
\usepackage[english]{babel}
\graphicspath{ {images/} }
\usepackage{blindtext}
\usepackage{fancyhdr}

\title{Age-Structured Epidemic Model under Vaccination with Vector Transmission}
\author{
Sourav Banerjee\footnote{Corresponding author:
    email: \texttt{sourav\_p220128ma@nitc.ac.in}, Department of Mathematics, NIT Calicut,
    Calicut 673601, India},
Thomas Götz\footnote{Mathematical Institute
University of Koblenz
Universitätsstr. 1
D-56070 Koblenz
Büro G 114
    \newline e-mail: \texttt{goetz@uni-koblenz.de}},
Satyananda Panda\footnote{Department of Mathematics, NIT Calicut,
    Calicut 673601, India \newline e-mail: \texttt{satyanand@nitc.ac.in}}
}
\begin{document}
\begin{abstract}
Dengue remains a major global public health concern 
due to its high morbidity and economic burden. 
Mathematical modeling is essential to understand 
its transmission mechanisms and for evaluating 
intervention strategies. In this paper, we formulate 
a vector--host model in which the human population 
is structured by age, and vaccinated individuals 
are further described by time since vaccination. 
The mosquito population is coupled to the host 
dynamics and reduced under a quasi-steady-state 
assumption. By integrating over vaccination age, 
we obtain a nonlinear steady-state formulation and 
express the endemic equilibrium as a fixed-point 
problem for the infected mosquito population. 
Using monotonicity arguments and the intermediate 
value theorem, we establish existence and uniqueness 
of the endemic equilibrium under a biologically 
interpretable threshold condition on the basic 
reproduction number. The analysis highlights the 
influence of age-dependent vaccination on long-term 
dengue dynamics.
\end{abstract}


\maketitle
\section{Introduction}
Dengue is one of the most rapidly spreading mosquito-borne viral diseases worldwide and remains a major public health concern in tropical and subtropical regions \citep{world2025report,Bhatt2013}. The disease is transmitted primarily by \emph{Aedes} mosquitoes and exhibits complex epidemiological patterns influenced by demographic structure, immunity development, environmental factors, and vector–host interactions. Epidemiological evidence suggests that dengue incidence and hospitalization rates depend strongly on patient age, indicating that demographic heterogeneity plays a crucial role in transmission dynamics.

Mathematical models have long been used to understand the spread and control of infectious diseases. Classical compartmental models based on ordinary differential equations, such as SIR and SEIR frameworks, provide valuable insights into threshold behavior, stability of equilibria, and control strategies \citep{Anderson1992,Hethcote2000}. For vector-borne diseases, host–vector models extending the Ross–Macdonald paradigm have been extensively investigated, leading to the formulation of the basic reproduction number and its role in determining disease persistence \citep{Smith2012,Esteva1998}. While these models capture essential biological mechanisms, they typically assume homogeneous host populations and do not account for age-dependent transmission effects.

Age-structured models formulated as transport-type partial differential equations offer a natural framework to incorporate demographic heterogeneity into epidemic modeling \citep{Webb1985,Inaba2017}. In such models, epidemiological parameters such as transmission, mortality, and recovery rates may vary with chronological age, yielding a more realistic description of disease spread. Extensions of the basic reproduction number to age-structured systems have been developed using next-generation operator theory and spectral methods \citep{Diekmann1990,Thieme2009}. These approaches provide rigorous criteria for the existence of disease-free and endemic equilibria in structured populations.

In the context of dengue transmission, age dependence is particularly relevant due to differences in exposure risk, immunity acquisition, and vaccination eligibility across age groups. Several studies have incorporated demographic effects and heterogeneous biting rates into host–vector models \citep{Feng1997}. The age-structured framework developed by Ganegoda \textit{et al}.~\citep{ganegoda2021age} investigated dengue transmission by structuring the host population with respect to chronological age and analyzing the resulting equilibrium behavior through an extension of the basic reproduction number. 

The present study extends the framework of Ganegoda \textit{et al}.~\citep{ganegoda2021age} by incorporating vaccination through an explicit vaccination-age structure and by coupling the age-structured human population with mosquito vector dynamics. The human population is partitioned into vaccinated and non-vaccinated subpopulations, and vaccine-induced protection is represented by a vaccination-age-dependent efficacy function that modulates susceptibility over time. This formulation permits the analysis of vaccination strategies within an age-structured vector–host model while preserving a rigorous mathematical treatment of the steady-state problem. We establish the existence and uniqueness of the steady-state infected mosquito population under a biologically interpretable threshold condition 
on the basic reproduction number, using 
monotonicity arguments and the intermediate value 
theorem. Numerical simulations are performed to 
investigate the impact of vaccination rate and 
vaccine efficacy on both vaccinated and 
non-vaccinated subpopulations. The results 
highlight the importance of age-structured 
vaccination strategies in controlling long-term 
dengue transmission.

The paper is structured as follows: The mathematical formulation of the age-structured vector–host model with vaccination is described in the next section. The reduction of the model, obtained by integrating over the vaccination-age variable, is presented in Sec. 3. The existence and uniqueness of the endemic equilibrium, based on a fixed-point formulation, are discussed in Sec. 4. The numerical method for computing the equilibrium solution is given in Sec. 5. The numerical results and their interpretation are presented in Sec. 6, and the paper closes with concluding remarks in the final section.

\section{An age-structured SISUV model}
We consider a human population structured by age $x \in [0,A]$ where $A$ is the maximum age in the population and time $t \ge 0$. Vaccinated individuals are additionally structured by vaccination age $y \ge 0$, representing the time elapsed since vaccination. The total population density at time $t$ and age $x$ is defined by
\begin{equation}\label{eq:pdef}
p(t,x) = Q(t,x) + J(t,x) + \int_0^\infty \Big[S(t,x,y) + I(t,x,y)\Big] dy,
\end{equation}
where $Q(t,x)$ and $J(t,x)$ denote the densities of non-vaccinated susceptible and infected individuals, respectively, and $S(t,x,y)$ and $I(t,x,y)$ denote the densities of vaccinated susceptible and infected individuals, respectively, with vaccination age $y$. The total host population is given by
\begin{equation}\label{eq:Nt}
N(t) := \int_0^A p(t,x)\,dx.
\end{equation}
The mosquito population is divided into susceptible mosquitoes $u(t)$ and infected mosquitoes $v(t)$, with total population $\mathit{M}(t) = u(t) + v(t)$. The non-vaccinated susceptible population evolves according to
\begin{equation}
\left\{
\begin{aligned}
\frac{\partial Q}{\partial t} + \frac{\partial Q}{\partial x} &=
- \frac{1}{N(t)} \vartheta(x) v(t) Q
- \mu(x) Q
+ \gamma J
- \nu_S Q, \\
Q(t, 0) &= \int_0^A b(s) p(t,s) ds,
\end{aligned}
\right.
\label{eq:nvs}
\end{equation}
where infection occurs through contact between susceptible humans and infected mosquitoes, with $\vartheta(x)$ denoting the age-dependent transmission rate from mosquitoes to humans, $\mu(x)$ denotes the age-dependent natural mortality rate assumed to be independent of the infection status, $\gamma$ is the recovery rate, and $\nu_S$ is the vaccination rate of susceptible individuals. We assume that all newborns are free of dengue infection and therefore enter exclusively into the non-vaccinated susceptible class with birth rate $b(x)$. The non-vaccinated infected population satisfies
\begin{equation}
\left\{
\begin{aligned}
\frac{\partial J}{\partial t} + \frac{\partial J}{\partial x} &=
\frac{1}{N(t)} \vartheta(x) v(t) Q
- \mu(x) J
- \gamma J
- \nu_I J, \\
J(t, 0) &= 0,
\end{aligned}
\right.
\label{eq:nvi}
\end{equation}
where $\nu_I$ represents the vaccination rate of infected individuals. Vaccinated susceptible individuals are structured by both chronological and vaccination age and satisfy
\begin{equation}
\left\{
\begin{aligned}
\frac{\partial S}{\partial t} + \frac{\partial S}{\partial x} + \frac{\partial S}{\partial y} &=
- [1 - \Omega(y)] \frac{1}{N(t)} \vartheta(x) v(t) S
- \mu(x) S
+ \gamma I, \\
S(t, x, 0) &= \nu_S Q(t,x),
\end{aligned}
\right.
\label{eq:vs}
\end{equation}
where $\Omega(y)$ denotes the vaccination-age dependent efficacy function, representing the level of protection conferred by vaccination age $y$.
Similarly, vaccinated infected individuals satisfy
\begin{equation}
\left\{
\begin{aligned}
\frac{\partial I}{\partial t} + \frac{\partial I}{\partial x} + \frac{\partial I}{\partial y} &=
[1 - \Omega(y)] \frac{1}{N(t)} \vartheta(x) v(t) S
- \mu(x) I
- \gamma I, \\
I(t, x, 0) &= \nu_I J(t,x),
\end{aligned}
\right.
\label{eq:vi}
\end{equation}
reflecting infection of vaccinated individuals with reduced susceptibility. 
The human host system is coupled with the mosquito population dynamics through the following equations. The susceptible mosquito population satisfies
\begin{equation}
\left\{
\begin{aligned}
\frac{du}{dt} &= 
\Lambda
- u(t) \int_0^A \sigma(x) J(t,x) \, dx
- u(t) \int_0^A \int_0^{\infty} \sigma(x) I(t,x,y) \, dy \, dx
- \rho u(t), \\
u(0) &= u_0,
\end{aligned}
\right.
\label{eq:sm}
\end{equation}
where $\Lambda$ denotes the mosquito birth rate and $\rho$ the natural death rate. The function $\sigma(x)$ represents the age-dependent transmission rate from humans to mosquitoes. The infected mosquito population evolves according to
\begin{equation}
\left\{
\begin{aligned}
\frac{dv}{dt} &= 
u(t) \int_0^A \sigma(x) J(t,x) \, dx
+ u(t) \int_0^A \int_0^{\infty} \sigma(x) I(t,x,y) \, dy \, dx
- \rho v(t), \\
v(0) &= v_0.
\end{aligned}
\right.
\label{eq:im}
\end{equation}
Here $u_0$ and $v_0$ are the initial mosquito population of susceptible and infected class respectively. In the next section, we reduce the full age–vaccination structured system by assuming steady-state mosquito dynamics and aggregating the vaccinated compartments to facilitate the analysis of endemic equilibria.
%=============================================
\section{Model Reduction}
 For simplicity, we assume that the non-vaccinated susceptible and non-vaccinated infected individuals share the same vaccination rate, that is, $\nu_S = \nu_I = \nu$. Upon assuming a sufficiently large time scale, the vector and human population are considered to exhibit minimal variation within a single unit of time. Under such steady-state conditions, the population is assumed to reflect the system's inherent nontrivial equilibrium.  Thus Eq.~\eqref{eq:Nt} reads $N(t)=N=\int_0^A p(x)dx$. 
 Moreover, vaccination age is intrinsically limited by human lifespan and natural mortality; consequently the population density corresponding to sufficiently large vaccination ages becomes negligible; therefore, we assume that   $S(x,y)$ and $I(x,y)$ vanish as $y \to \infty$. To reduce the dimensionality of the vaccinated subsystem, we integrate equations \eqref{eq:vs} and \eqref{eq:vi} with respect to vaccination age $y$ over $[0,\infty)$. This yields
\begin{equation}\label{eq:vsr}
\int_0^\infty  \frac{\partial S}{\partial x}\,dy -S(x,y=0)
= \int_0^\infty \left [ - [1 - \Omega(y)] \frac{1}{N} \vartheta(x) \,v^{\ast} S - \mu(x) S + \gamma I \right ]dy,
\end{equation}
\begin{equation}\label{eq:vir}
\int_0^\infty  \frac{\partial I}{\partial x}\,dy-I(x,y=0)
= \int_0^\infty \left [  [1 - \Omega(y)] \frac{1}{N} \vartheta(x)\,v^{\ast} S - \mu(x) I - \gamma I \right ]dy.
\end{equation}
Here $v^{\ast}$ is steady state infected mosquito population. Vaccination is not vertically transmitted from mother to child at birth. Therefore,
\begin{equation}\label{eq:si0}
S(x=0,y) = I(x=0,y) = 0.
\end{equation}
Moreover, children are generally vaccinated only after the age of four years~\citep{kling2023stiko}.
Next, on adding the steady-state form of Eqs.~\eqref{eq:nvs}, \eqref{eq:nvi} with \eqref{eq:vsr}, and \eqref{eq:vir}, we get
\begin{equation}\label{eq:std_nsi}
\frac{d p}{d x}=-\mu(x)p(x), \,\,\,\,\,p(x=0)=\int_0^A b(s) \, p(s) \, ds,
\end{equation}
where the total population of age $x$ is denoted by $p$ (see Eq.~\eqref{eq:pdef}). Next, Eq.~\eqref{eq:std_nsi} is a linear ODE, and the solution is given by
\begin{equation}\label{eq:p}
p(x) = p_0 e^{-\int_0^x \mu(s)\, ds}.
\end{equation}
 Here the integrating constant $p_0$ is defined as \(p_0 = \int_0^A b(s)\, p(s)\, ds\). Next, we define the aggregated vaccinated populations by
\[
\hat{S}(x) := \int_0^\infty S(x,y)\,dy,
\qquad
\hat{I}(x) := \int_0^\infty I(x,y)\,dy,
\]
where $\hat{S}(x)$ and $\hat{I}(x)$ denote the total densities of vaccinated susceptible and vaccinated infected individuals, respectively, of chronological age $x$, irrespective of vaccination age. We then define the total non-vaccinated population by $\tilde{p}$ and the total vaccinated population by $\tilde{q}$, where
\begin{align}
\tilde{p}(x) &:= Q(x) + J(x), \label{eq:tillp}\\
\tilde{q}(x) &:= \hat{S}(x) + \hat{I}(x). \label{eq:tillq}
\end{align}
It follows that
\begin{equation}\label{eq:pp}
\tilde{p}(x) + \tilde{q}(x) = p(x).
\end{equation}
 Hence, by adding Eq.~\eqref{eq:nvs} and Eq.~\eqref{eq:nvi} we have
\begin{equation}\label{eq:ptil}
\frac{d \tilde{p}}{dx} = -\left(\mu + \nu\right) \tilde{p}
, \quad \tilde{p}(x=0) = p_0
\end{equation}
and by adding Eq.~\eqref{eq:vsr} and Eq.~\eqref{eq:vir} we get 
\begin{equation}\label{eq:qtil}
\frac{d \tilde{q}}{dx} = \nu \tilde{p} - \mu \tilde{q}
, \quad \tilde{q}(x=0) = 0
\end{equation}
Using Eq.~\eqref{eq:p}, the solutions of Eqs.~\eqref{eq:ptil} and~\eqref{eq:qtil} are given by
\begin{align}
\tilde{p}(x) &= p(x)\, e^{-\nu x}, \label{eq:p_til}\\
\tilde{q}(x) &= p(x)\left(1 - e^{-\nu x}\right). \label{eq:q_til}
\end{align}
Next, we introduce the following age-dependent population fractions:
\begin{equation}\label{eq:wxyz}
\begin{aligned}
W(x) &:= \frac{J(x)}{p(x)}, 
&\qquad
Z(x) &:= \frac{Q(x)}{p(x)},\\
X(x) &:= \frac{\hat{I}(x)}{p(x)}, 
&\qquad
Y(x) &:= \frac{\hat{S}(x)}{p(x)}.
\end{aligned}
\end{equation}
Since all population densities are nonnegative and $p(x)>0$, it follows that
\begin{equation}\label{eq:wxyz_bounds}
0 \le W(x),\, Z(x),\, X(x),\, Y(x) \le 1, 
\qquad
W(x)+Z(x)+X(x)+Y(x)=1.
\end{equation}
Here, $W(x)$ and $Z(x)$ denote the compartments of non-vaccinated infected and non-vaccinated susceptible individuals of age $x$, respectively, while $X(x)$ and $Y(x)$ represent the compartments of vaccinated infected and vaccinated susceptible individuals of age $x$, respectively. Moreover, using definitions~\eqref{eq:tillp} and~\eqref{eq:tillq} together with expressions~\eqref{eq:p_til}, \eqref{eq:q_til}, and~\eqref{eq:wxyz}, we obtain
\begin{equation}\label{eq:ZYrel}
Z(x)=e^{-\nu x}-W(x),
\qquad
Y(x)=1-e^{-\nu x}-X(x).
\end{equation}
At steady state, Eqs.~\eqref{eq:sm} and~\eqref{eq:im} yield
\begin{equation}\label{eq:star}
u^{*} = \frac{\Lambda}{\Upsilon + \rho},
\qquad
v^{\ast} = \frac{\Lambda}{\Upsilon + \rho}\,\frac{\Upsilon}{\rho}
      = \frac{\Lambda \Upsilon}{\rho^{2}}
        \left(1 + \frac{\Upsilon}{\rho}\right)^{-1}.
\end{equation}
where \(\Upsilon\) denotes the total infection pressure exerted by infected humans on susceptible mosquitoes and is defined by 
\begin{equation}\label{eq:Upsi}
\Upsilon = \int_0^A \sigma(x) J(x) \, dx + \int_0^A \int_0^{\infty} \sigma(x) I(x,y) \, dy \, dx. 
\end{equation}
The terms $u^*$ and $v^{\ast}$ are the steady-state representation of mosquito populations. Let us introduce the following notation:
\begin{equation}\label{eq:sigt}
\tilde{\sigma} = \frac{ \displaystyle\int_0^A \sigma(x) J(x) \, dx + \int_0^A \int_0^{\infty} \sigma(x) I(x, y) \, dy \, dx }{ \displaystyle\int_0^A \Bigg[J(x)+ \int_0^{\infty} I(x,y) dy\Bigg]dx}
\end{equation}
where the quantity $\tilde{\sigma}$ represents the average transmission rate from infected humans to mosquitoes, weighted by the densities of infected individuals across age and vaccination age. Due to isolation measures and the clean hospital environment, the infection pressure exerted by humans on mosquitoes is assumed to be small relative to the mosquito mortality rate, i.e., $\Upsilon/\rho \ll 1$. This assumption reflects relatively low human-to-mosquito transmission compared to mosquito mortality, consistent with controlled or hospital-based environments. Under this assumption, Eqs.~\eqref{eq:star} and~\eqref{eq:sigt} yield
\begin{align}
v^{\ast} &\simeq \frac{\Lambda \Upsilon}{\rho^{2}}
      = \frac{\Lambda \tilde{\sigma}}{\rho^{2}}
        \int_{0}^{A} \Big[J(x) + \hat{I}(x)\Big]\, dx \nonumber \\
     &= \frac{\Lambda \tilde{\sigma}}{\rho^{2}}
        \int_{0}^{A} \big[ W(x) + X(x) \big]\, p(x)\, dx.
\label{eq:vstar}
\end{align}
Next, evaluating Eqs.~\eqref{eq:nvi} and~\eqref{eq:std_nsi} at steady state, together with Eqs.~\eqref{eq:wxyz} and~\eqref{eq:ZYrel}, we obtain
\begin{equation}\label{eq:Wode}
W'(x)+(\gamma+\nu)W(x)=\frac{\vartheta(x)}{N}v^{\ast}\left(e^{-\nu x}-W(x)\right),
\qquad
W(0)=0.
\end{equation}
Analogously, the vaccinated infected individuals satisfy Eq.~\eqref{eq:vir}, where the first contribution on the right-hand side accounts for vaccine-induced protection and is approximated as
\begin{equation}\label{eq:aprx}
\int_0^\infty (1-\Omega(y))\, S(x,y)\, dy = (1-\Omega_b)\, \hat{S}(x).
\end{equation}
Here, $\Omega_b$ denotes the averaged vaccine efficacy over the vaccination-age distribution.
Using Eqs.~\eqref{eq:vir}, \eqref{eq:std_nsi}, \eqref{eq:wxyz}, and~\eqref{eq:ZYrel} together with the approximation~\eqref{eq:aprx}, we obtain
\begin{equation}\label{eq:Xode}
X'(x) + \gamma X(x)
= \frac{\vartheta(x)}{N} \, v^{\ast} (1-\Omega_b)
\bigl( (1 - e^{-\nu x}) - X(x) \bigr)
+ \nu W(x),
\qquad
X(0) = 0.
\end{equation}
The coupled system for $W$ and $X$ given by Eqs.~\eqref{eq:Wode} and~\eqref{eq:Xode} depends on the mosquito population $v^{\ast}$, which is itself determined by $W$ and $X$ through Eq.~\eqref{eq:vstar}. The final coupled system is written as follows:
\begin{equation}\label{eq:fs}
\left\{
\begin{aligned}
W'(x) + (\gamma + \nu) W(x)
&= \frac{\vartheta(x)}{N} \, v^{\ast} \left( e^{-\nu x} - W(x) \right),
\qquad W(0) = 0, \\[6pt]
X'(x) + \gamma X(x)
&= \frac{\vartheta(x)}{N} \, v^{\ast} (1 - \Omega_b)
\left[ (1 - e^{-\nu x}) - X(x) \right]
+ \nu W(x),
\qquad X(0) = 0, \\[6pt]
v^{\ast}
&= \frac{\Lambda \tilde{\sigma}}{\rho^2}
\int_{0}^{A} \bigl[ W(x) + X(x) \bigr] \, p(x) \, dx.
\end{aligned}
\right.
\end{equation}
The mutual dependence in Eqs.~\eqref{eq:fs} leads to a nonlinear fixed-point problem. In the next section, we establish existence and uniqueness of the corresponding solution.
%=======================================
\section{Existence and Uniqueness of the Endemic 
Equilibrium}\label{sec:EU}
In this section, we establish the existence and uniqueness of the endemic equilibrium. 
Throughout, we assume:
\begin{itemize}
\item[(A1)] $\vartheta \in L^\infty(0,A)$ with $\vartheta(x) > 0$,
\item[(A2)] $p \in C([0,A])$ with $p(x) > 0$ for all $x \in [0,A]$,
\item[(A3)] $0 \le \Omega_b \le 1$.
\end{itemize}
\noindent
For each fixed $v^{\ast} \ge 0$, the solutions of system~\eqref{eq:fs} are given by
\begin{equation}\label{eq:WXvc}
\begin{cases}
\displaystyle
W_{v^{\ast}}(x)
=
\frac{v^{\ast}}{N}
\int_0^x
\vartheta(s)\,
\exp\!\left(
-\nu s-(\gamma+\nu)(x-s)
-\frac{v^{\ast}}{N}\int_s^x\vartheta(r)\,dr
\right)ds,
\\[10pt]
\displaystyle
X_{v^{\ast}}(x)
=
\int_0^x
\left[
\frac{\vartheta(s)}{N}v^{\ast}(1-\Omega_b)\left(1-\exp(-\nu s)\right)
+\nu W_{v^{\ast}}(s)
\right]
\exp\!\left(
-\gamma(x-s)
-\frac{v^{\ast}(1-\Omega_b)}{N}
\int_s^x\vartheta(r)\,dr
\right)ds.
\end{cases}
\end{equation}
Next, we introduce two functions $w(x)$ and $\xi(x)$, which represent the marginal infection densities per unit infected mosquito population in the limit of low transmission:
\begin{equation}\label{eq:wxi}
w(x)=\lim_{v^{\ast}\to 0^+}\frac{W_{v^{\ast}}(x)}{v^{\ast}}, \qquad 
\xi(x)=\lim_{v^{\ast}\to 0^+}\frac{X_{v^{\ast}}(x)}{v^{\ast}}.
\end{equation}
The integrands in Eq.~\eqref{eq:WXvc} are uniformly bounded under assumptions (A1)--(A3). Hence, by the dominated convergence theorem, the limits in Eq.~\eqref{eq:wxi} exist and are given by
\begin{align}
w(x)
&=
\frac{1}{N}\int_0^x
\vartheta(s)\,
\exp\!\left(-\nu s-(\gamma+\nu)(x-s)\right)\,ds,
\label{eq:wlim}\\[6pt]
\xi(x)
&=
\int_0^x
\left[
\frac{\vartheta(s)}{N}(1-\Omega_b)\left(1-\exp(-\nu s)\right)
+\nu\,w(s)
\right]
\exp\!\left(-\gamma(x-s)\right)\,ds.
\label{eq:xilim}
\end{align}
In Eqs.~\eqref{eq:wlim} and \eqref{eq:xilim}, since all factors in the integrands are nonnegative, $w(x)\ge 0$ and $\xi(x)\ge 0$ for all
$x\in[0,A]$.
We define the basic reproduction number
\begin{equation}\label{eq:R0}
\mathcal{R}_0
:=
\frac{\Lambda\tilde{\sigma}}{\rho^2}
\int_0^A p(x)\bigl(w(x)+\xi(x)\bigr)\,dx.
\end{equation}
This quantity represents the expected number 
of secondary mosquito infections arising from 
a single infected mosquito through one complete 
transmission cycle in a fully susceptible 
population. Next, we define the operator and the bound
\begin{equation}\label{eq:Fop}
F(v^{\ast})
:=
\frac{\Lambda\tilde{\sigma}}{\rho^2}
\int_0^A p(x)\bigl(W_{v^{\ast}}(x)+X_{v^{\ast}}(x)\bigr)\,dx \qquad \qquad M := \frac{\Lambda\tilde{\sigma}}{\rho^2}\cdot\,A\,\|p\|_{\infty}\,(2+\nu A)
\end{equation}
\begin{theorem}\label{thm:main}
(Existence and uniqueness of the endemic equilibrium).
\begin{enumerate}
\item[\emph{(i)}] If $\mathcal{R}_0 \le 1$, then $v^{\ast}=F(v^{\ast})$
has no solution in $(0,M]$.
\item[\emph{(ii)}] If $\mathcal{R}_0 > 1$, then there exists a
unique $v^{\ast}\in(0,M]$ such that $v^{\ast}=F(v^{\ast})$.
\end{enumerate}
\end{theorem}
\begin{proof}
We first verify that $F$ maps $(0,M]$ into $(0,M]$. From Eq.~\eqref{eq:WXvc}, noting that all exponential factors are bounded above by 1 and observing
\begin{equation}\nonumber
\frac{v^{\ast}}{N}\int_0^x\vartheta(s)
\exp\!\left(-\frac{v^{\ast}}{N}\int_s^x
\vartheta(r)\,dr\right)ds
= 1-e^{-T}\le 1,
\quad T=\frac{v^{\ast}}{N}\int_0^x\vartheta(r)\,dr\ge 0,
\end{equation}
it follows from the above representation that
\begin{equation}\nonumber
0 < W_{v^{\ast}}(x)\le 1, \quad 
0\leq X_{v^{\ast}}(x)\le 1+\nu A,
\quad\text{uniformly in }v^{\ast}> 0.
\end{equation}
Hence, for all $x\in[0,A]$:
\begin{equation}\nonumber
0 \leq W_{v^{\ast}}(x)+X_{v^{\ast}}(x)\le 2+\nu A.
\end{equation}
Thus,
\begin{equation}
0 \leq F(v^{\ast})
=
\frac{\Lambda\tilde{\sigma}}{\rho^2}
\int_0^A p(x)
\bigl(W_{v^{\ast}}(x)+X_{v^{\ast}}(x)\bigr)\,dx
\le
\frac{\Lambda\tilde{\sigma}}{\rho^2}\cdot A\,\|p\|_{\infty} (2+\nu A)
= M
\end{equation}
Hence, $F$ maps $(0,M]$ into $(0,M]$. Next, we define
\begin{equation}\label{eq:def_G}
G(v^{\ast}) := \frac{F(v^{\ast})}{v^{\ast}} - 1.
\end{equation}
We observe that $F$ is continuous on $(0,M]$. Indeed, under assumptions (A1)–(A3), for each $x \in [0,A]$, the functions $W_{v^{\ast}}(x)$ and $X_{v^{\ast}}(x)$ depend continuously on $v^{\ast}$, and the integrand in $F(v^{\ast})$ is uniformly bounded by integrable function on $[0,A]$. Hence, by the dominated convergence theorem, $F$ is continuous. Consequently, $G$ is continuous on $(0,M]$.
Therefore, from the definition~\eqref{eq:def_G}, it follows that:
\begin{equation}\label{eq:GM}
G(M) = \frac{F(M)}{M}-1 \le 0,\, \quad
G(0^+)
=
\lim_{v^{\ast}\to 0^+}\frac{F(v^{\ast})}{v^{\ast}}-1
=
\frac{\Lambda\tilde{\sigma}}{\rho^2}
\int_0^A p(x)\bigl(w(x)+\xi(x)\bigr)\,dx-1
=
\mathcal{R}_0-1.
\end{equation}
Moreover, in order to show the monotonicity of $G$, we compute $G'(v^{\ast})$ and the partial derivative terms 
$\partial(W_{v^*}/v^*)/\partial v^*$,
$\partial(X_{v^*}/v^*)/\partial v^*$ need to be  bounded 
uniformly. We first calculate the corresponding partial derivative terms as follows:
\begin{equation}\label{eq:wxp}
\frac{\partial}{\partial v^*}\left(\frac{W_{v^*}}{v^*}\right) =\mathcal{T}_1(x), \qquad
\frac{\partial}{\partial v^*}\left(\frac{X_{v^*}}{v^*}\right)
=\mathcal{T}_2(x)+\mathcal{T}_3(x)+\mathcal{T}_4(x)
\end{equation}
where 
\begin{equation}\label{eq:all_T}
\begin{aligned}
\mathcal{T}_1(x)
&=-\frac{1}{N^2}
\int_0^x
\vartheta(s)\,
\exp\!\left(
-\nu s-(\gamma+\nu)(x-s)
-\frac{v^{\ast}}{N}\int_s^x\vartheta(r)\,dr
\right)
\left(\int_s^x\vartheta(r)\,dr\right)ds,\\
\mathcal{T}_2(x)
&:=
-\int_0^x
\frac{\vartheta(s)}{N}(1-\Omega_b)\left(1-\exp(-\nu s)\right)
\exp\!\left(
-\gamma(x-s)
-\frac{v^{\ast}(1-\Omega_b)}{N}
\int_s^x\vartheta(r)\,dr
\right)
\frac{(1-\Omega_b)}{N}
\int_s^x\vartheta(r)\,dr\,ds,
\\[6pt]
\mathcal{T}_3(x)
&:=
\int_0^x
\nu\,\mathcal{T}_1(s)\,
\exp\!\left(
-\gamma(x-s)
-\frac{v^{\ast}(1-\Omega_b)}{N}
\int_s^x\vartheta(r)\,dr
\right)ds,
\\[6pt]
\mathcal{T}_4(x)
&:=
-\int_0^x
\nu\,\frac{W_{v^{\ast}}(s)}{v^{\ast}}\,
\exp\!\left(
-\gamma(x-s)
-\frac{v^{\ast}(1-\Omega_b)}{N}
\int_s^x\vartheta(r)\,dr
\right)
\frac{(1-\Omega_b)}{N}
\int_s^x\vartheta(r)\,dr\,ds.
\end{aligned}
\end{equation}
\noindent
Consequently, we have
\[
\mathcal{T}_1(x) < 0, \quad
\mathcal{T}_2(x) \le 0, \quad
\mathcal{T}_3(x) < 0, \quad
\mathcal{T}_4(x) \le 0.
\]
Moreover, by bounding each factor in Eq.~\eqref{eq:all_T}, using $\vartheta\in L^\infty(0,A)$ from (A1) and $x\le A$, we obtain
\begin{align*}
|\mathcal{T}_1(x)|
&\le
\frac{\|\vartheta\|_\infty^2 A^2}{N^2}
:= C_1, 
\qquad
|\mathcal{T}_2(x)|
\le C_1,\\
|\mathcal{T}_3(x)|
&\le
\nu A\, C_1
:= C_2, 
\qquad
|\mathcal{T}_4(x)|
\le C_2.
\end{align*}
Hence each 
$\mathcal{T}_i$ satisfies
\begin{equation}\label{eq:Ti_bound}
-C^{\ast} \le \mathcal{T}_i(x)\le 0,
\quad i=1,2,3,4,
\end{equation}
where $C^{\ast}:=\max\{C_1,C_2\}$ is finite and independent of $v^{\ast}$.
 Therefore, the partial derivatives in Eq.~\eqref{eq:wxp} are uniformly bounded in $v^{\ast}$ under assumptions (A1)--(A3). Consequently, the Leibniz rule for differentiation under the integral sign applies, yielding
\begin{equation}\label{eq:Gp}
\begin{aligned}
G'(v^{\ast})
&=
\frac{\Lambda\tilde{\sigma}}{\rho^2}
\int_0^A p(x)
\left(
\frac{\partial}{\partial v^{\ast}}\!\left(\frac{W_{v^{\ast}}}{v^{\ast}}\right)
+
\frac{\partial}{\partial v^{\ast}}\!\left(\frac{X_{v^{\ast}}}{v^{\ast}}\right)
\right)dx\\
&=\frac{\Lambda\tilde{\sigma}}{\rho^2}
\int_0^A p(x)
\left(  \mathcal{T}_1+\mathcal{T}_2+\mathcal{T}_3+\mathcal{T}_4 \right)dx
\end{aligned}
\end{equation}
Hence from Eqs.~\eqref{eq:all_T} and \eqref{eq:Gp}, we have 
$G'(v^{\ast})<0$ i.e., $G$ is strictly decreasing
on $(0,M]$.\\
\medskip
\noindent\emph{Proof of~(i).}\\
By Eq.~\eqref{eq:GM}, we know $G(0^+)=\mathcal{R}_0-1\le 0$. Since $G$ is strictly decreasing, it follows that:
\[
G(v^{\ast}) < G(0^+) = \mathcal{R}_0-1 \le 0
\quad\text{for all }v^{\ast}\in(0,M].
\]
Hence $G(v^{\ast})=0$ has no solution in $(0,M]$, and no
endemic equilibrium exists.\\
\medskip
\noindent\emph{Proof of~(ii).}\\
Using Eq.~\eqref{eq:GM}, we have:
\[
G(0^+) = \mathcal{R}_0-1 > 0 \ge G(M). 
\]
Since $G$ is continuous and strictly decreasing on $(0,M]$
with $G(0^+)>0\ge G(M)$, the intermediate value property
guarantees the existence of at least one $v^{\ast}\in(0,M]$
satisfying $G(v^{\ast})=0$, that is, $v^{\ast}=F(v^{\ast})$. Strict
monotonicity of $G$ ensures uniqueness. Consequently,
the endemic equilibrium is unique.
\end{proof}
%========================================
%====================================
\section{Numerical Method}
We numerically compute the age distributions of the 
non-vaccinated and vaccinated populations, namely, 
$Q(x)$, $J(x)$, $\hat{S}(x)$, and $\hat{I}(x)$, 
together with the corresponding infected mosquito 
equilibrium $v^{\ast}$. These quantities are obtained 
from the population fractions $W(x)$, $X(x)$, $Y(x)$, 
and $Z(x)$ through the relations introduced in 
Eq.~\eqref{eq:wxyz}. The problem is described by a 
coupled system consisting of age-dependent ordinary 
differential equations for $W(x)$ and $X(x)$ given 
in Eqs.~\eqref{eq:Wode} and~\eqref{eq:Xode}, together 
with a fixed-point relation for $v^{\ast}$ in 
Eq.~\eqref{eq:vstar}. By the existence and uniqueness 
result established in Sec.~\ref{sec:EU}, this 
fixed-point problem admits a unique solution 
$v^{\ast}\in(0,M]$ whenever $\mathcal{R}_0>1$. 
Starting from an initial guess $v^{\ast(0)}>0$, 
the ODE system for $W_{v^{\ast(n)}}(x)$ and 
$X_{v^{\ast(n)}}(x)$ is solved on the interval $[0,A]$. 
The age variable is discretized using a uniform grid, 
and the system is advanced in age by a first-order 
explicit Euler scheme. After computing 
$W_{v^{\ast(n)}}$ and $X_{v^{\ast(n)}}$, the mosquito 
equilibrium is updated according to
\begin{equation}\label{eq:iter}
v^{\ast(n+1)}
=
\frac{\Lambda\tilde{\sigma}}{\rho^2}
\int_0^A p(x)
\bigl(W_{v^{\ast(n)}}(x)+X_{v^{\ast(n)}}(x)\bigr)\,dx,
\end{equation}
where the integral is approximated using the composite 
trapezoidal rule. The iteration is repeated until the 
convergence criterion
\[
\max\left\{
\|W^{(n+1)}-W^{(n)}\|_{\infty},\;
\|X^{(n+1)}-X^{(n)}\|_{\infty},\;
|v^{\ast(n+1)}-v^{\ast(n)}|
\right\}
< \varepsilon
\]
is satisfied for a prescribed tolerance $\varepsilon>0$. 
Once convergence is achieved, the remaining components 
$Q(x)$, $J(x)$, $\hat{S}(x)$, and $\hat{I}(x)$ are 
recovered from Eq.~\eqref{eq:ZYrel}.

\section{Results and Discussion}
For the numerical simulations, we consider a maximal age $A = 80$ years with a constant human mortality rate $\mu = 0.01~\text{year}^{-1}$ and a total population size $N = 550{,}670$ persons, corresponding to a representative urban population. The recovery rate is set to $\gamma = 73~\text{year}^{-1}$ (equivalently, $0.2~\text{day}^{-1}$), reflecting a mean infectious period of approximately five days, and the vaccination rate is taken as $\nu = 0.1~\text{year}^{-1}$. 

For the mosquito population, the recruitment rate is chosen as $\Lambda = 2704~\text{mosquitoes}~\text{year}^{-1}$ and the natural death rate as $\rho = 26~\text{year}^{-1}$, corresponding to an average mosquito lifespan of approximately 14 days, consistent with field observations for \emph{Aedes} mosquitoes. The average transmission rate from infected humans to mosquitoes is taken as $\tilde{\sigma} = 0.9~(\text{persons}\cdot\text{year})^{-1}$. The age-dependent transmission function $\vartheta(x)$, representing the transmission rate from infected mosquitoes to humans, is assumed to have units $(\text{mosquito}\cdot\text{year})^{-1}$ and is taken to be piecewise constant:
\[
\vartheta(x) =
\begin{cases}
0.5, & 0 \le x < 10, \\ 
0.3, & 10 \le x < 50, \\
0.1, & 50 \le x \le A,
\end{cases}
\]
reflecting higher exposure among younger individuals and reduced transmission at older ages. For the chosen parameter values, the basic reproduction number evaluates to $\mathcal{R}_0 = 1.5768$ for $\Omega_b = 0.9$ and $\mathcal{R}_0 = 2.9601$ for $\Omega_b = 0.5$, confirming that $\mathcal{R}_0 > 1$ in both cases. Hence, the conditions of Theorem-\ref{thm:main} are satisfied, and a unique endemic equilibrium exists in each scenario.

Figs.~\ref{fig:f} and~\ref{fig:f2} illustrate the age profiles of the non-vaccinated and vaccinated populations, respectively. In particular, the right panels show that increasing vaccine efficacy leads to a marked reduction in the densities of infected individuals, both in the non-vaccinated class $J(x)$ and in the vaccinated class $\hat{I}(x)$. Conversely, the left panels indicate an increase in the densities of susceptible individuals, namely the non-vaccinated susceptibles $Q(x)$ and the vaccinated susceptibles $\hat{S}(x)$, as vaccine efficacy increases.
\begin{figure}
\centering
\begin{subfigure}{0.48\linewidth}
    \includegraphics[width=\linewidth]{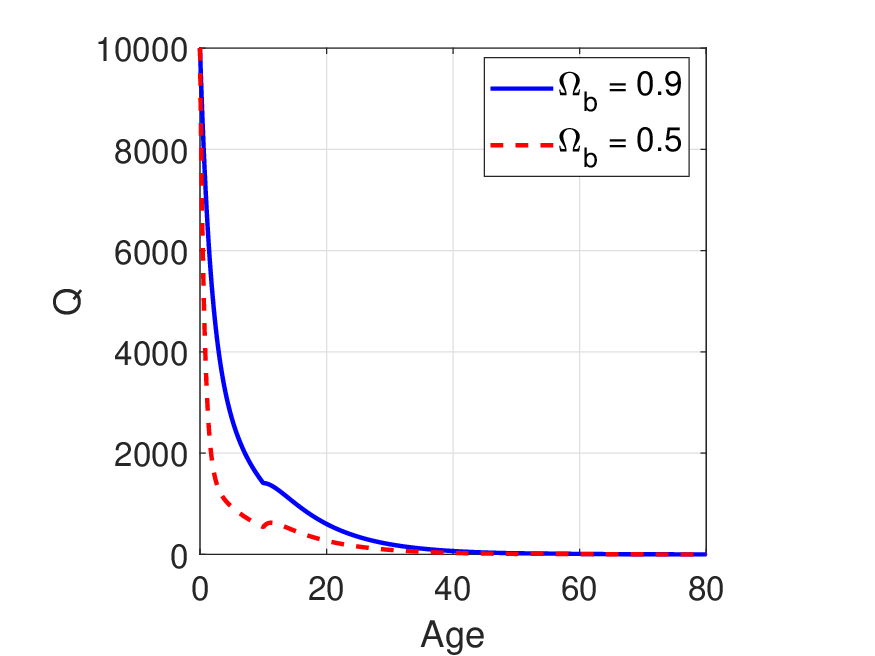}
    \caption{Non-vaccinated susceptible}
\end{subfigure}
\hfill
\begin{subfigure}{0.48\linewidth}
    \includegraphics[width=\linewidth]{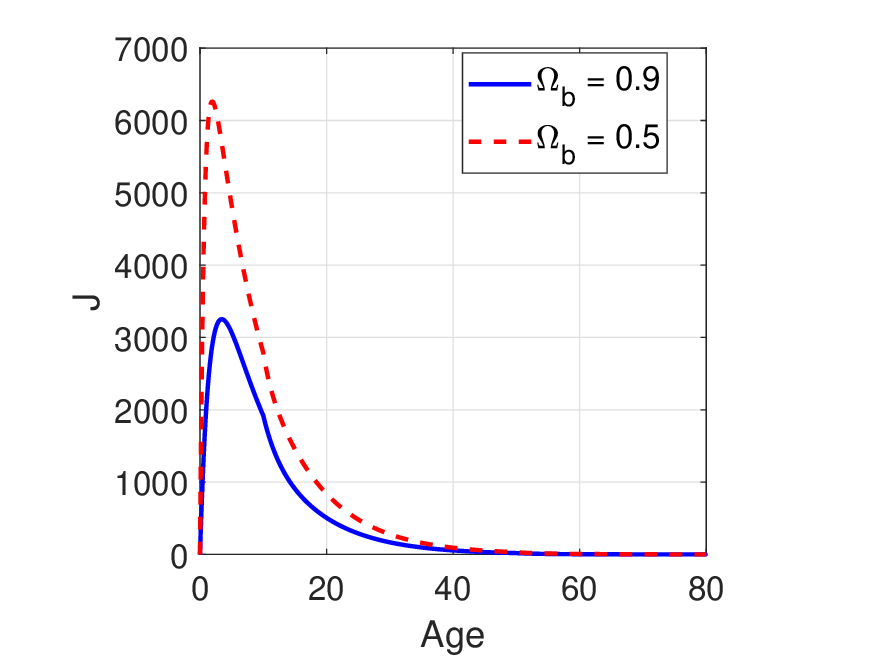}
    \caption{Non-vaccinated infected}
\end{subfigure}
\caption{Age profiles of non-vaccinated individuals 
for $\Omega_b=0.9$ and $\Omega_b=0.5$.}
\label{fig:f}
\end{figure}
%--------------------------------------------------
\begin{figure}
\centering
\begin{subfigure}{0.48\linewidth}
    \includegraphics[width=\linewidth]{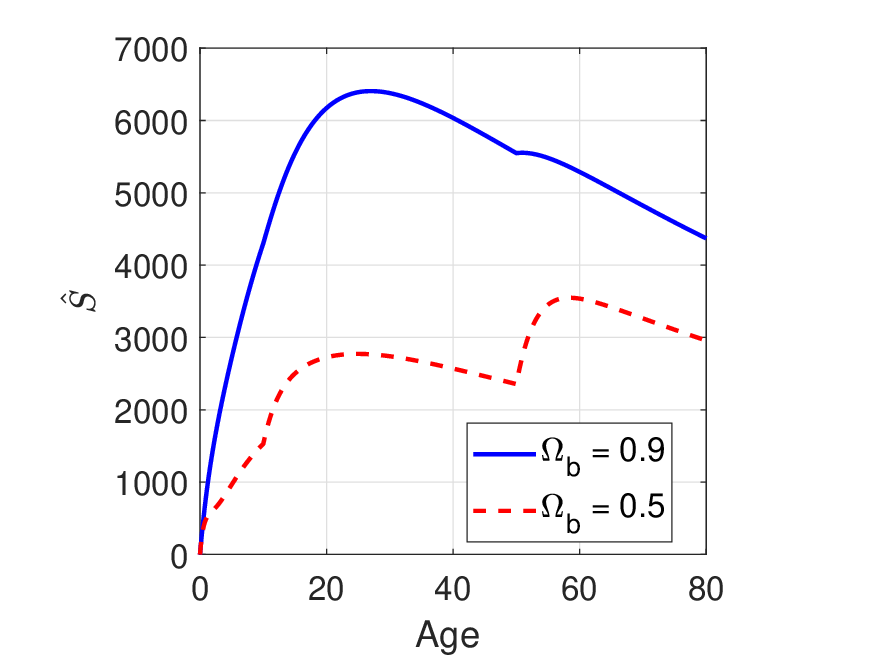}
    \caption{Vaccinated susceptible}
\end{subfigure}
\hfill
\begin{subfigure}{0.48\linewidth}
    \includegraphics[width=\linewidth]{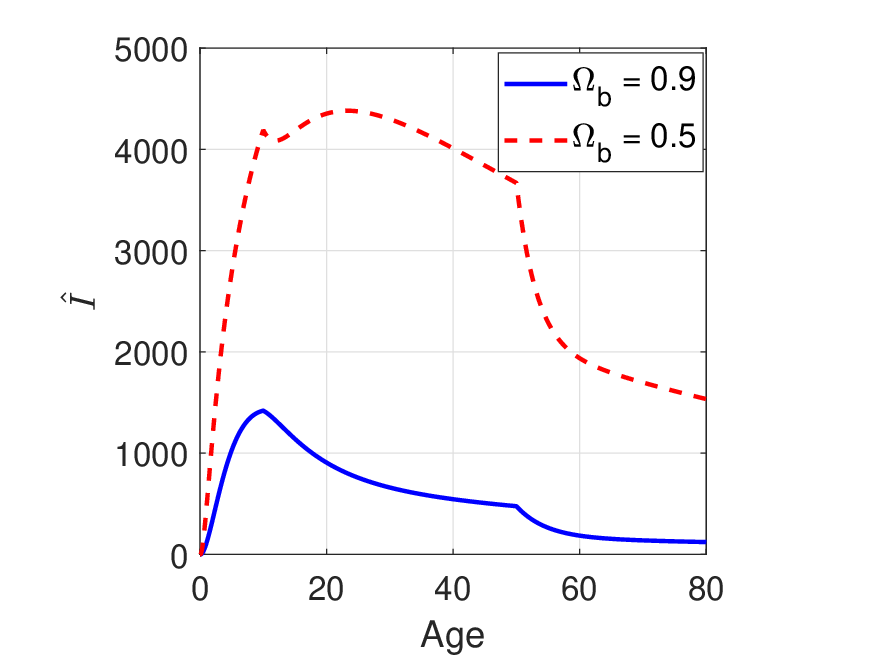}
    \caption{Vaccinated infected}
\end{subfigure}
\caption{Age profiles of vaccinated individuals 
for $\Omega_b=0.9$ and $\Omega_b=0.5$.}
\label{fig:f2}
\end{figure}
%===========================================================
\section{Conclusion}
In this work, we developed an age-structured 
vector--host model for dengue transmission that 
incorporates vaccination with explicit dependence 
on time since vaccination. By integrating over the vaccination-age variable, the model is reduced to an age-dependent system, and the endemic equilibrium is characterized as a fixed-point problem for the infected mosquito population. Existence and uniqueness of the endemic equilibrium are then established under a biologically meaningful threshold condition in terms of the basic reproduction number, using monotonicity arguments and the intermediate value theorem. The analysis demonstrates how vaccination and 
vaccination-age-dependent efficacy reduce endemic 
infection levels by lowering the susceptible human 
population, which in turn reduces the infection 
pressure on mosquitoes. Numerical simulations support the 
theoretical findings and illustrate the role of 
vaccine efficacy in reducing infection levels 
across age groups. The framework provides a 
mathematically rigorous basis for studying 
age-targeted vaccination strategies in 
vector-borne diseases and can be extended to 
include more detailed vector dynamics or waning 
immunity mechanisms.
\section*{Declaration of competing interest}
The authors declare that they have no known competing financial interests or personal relationships that could have appeared to influence the work reported in this paper.

\section*{Data availability statement}
Data will be made available on request

\section*{Acknowledgement}
The authors gratefully acknowledge the financial support provided under the DAAD--SPARC Project ``MAIHRT--Mathematical AI in Healthcare: Research and Teaching Perspective'' (Project ID: 57807795). This support enabled collaborative research between the National Institute of Technology Calicut and the University of Koblenz, facilitating research visits, exchange of ideas, and the development of the results presented in this work. 
%=======================================
\bibliographystyle{ieeetr}
\bibliography{ref}
\end{document}